\documentclass[12pt]{article}
\usepackage{amsmath,amssymb,amsthm,epsfig}
\usepackage{graphicx}

\renewcommand\a{{\bf a}}
\renewcommand\b{{\bf b}}

\renewcommand\d{{\bf d}}

\newcommand\eps{\epsilon}
\newcommand\eref[1]{$(\ref{#1})$}
\renewcommand\l{{\bf l}}
\newcommand\R{{\bf R}}
\renewcommand\r{{\bf r}}
\newcommand\s{{\bf s}}

\newcommand\rank{\mathop{\rm rank}}
\renewcommand\top{\mathop{\rm top}}
\renewcommand\u{{\bf u}}
\renewcommand\v{{\bf v}}
\newcommand\x{{\bf x}}

\newcommand\bz{{\bf 0}}
\newtheorem{theorem}{Theorem}
\newtheorem{lemma}[theorem]{Lemma}
\newcommand\fopt{f^{\rm opt}}
\newcounter{algstmt}
\newcommand\algnumber[1]{\refstepcounter{algstmt}$\langle\thealgstmt\rangle$\label{#1}}
\newcommand\stmtref[1]{$\langle\ref{#1}\rangle$}
\title{Nonnegative Matrix Factorization via Rank-One Downdate}
\author{Michael Biggs\thanks{Department of Statistics, University
of Waterloo, 200 University Ave.~W., Waterloo, Ontario, Canada
N2L 3G1, {\tt mike@doubleplum.net}.} \and
Ali Ghodsi\thanks{Department of Statistics, University
of Waterloo, 200 University Ave.~W., Waterloo, Ontario, Canada
N2L 3G1, {\tt aghodsib@uwaterloo.ca.}} \and
Stephen Vavasis\thanks{Department of Combinatorics and Optimization, 
University of Waterloo, 200 University Ave.~W., Waterloo, Ontario, Canada
N2L 3G1, {\tt vavasis@math.uwaterloo.ca.}}}
\begin{document}
\maketitle
\begin{center}
This preliminary version of the manuscript still has an incomplete
literature review and is missing the section on computational
testing.  It contains a complete proof of the main theorem.
Please check back here after June 15, 2008
for a more complete version of this manuscript.
\end{center}
\begin{abstract}
Nonnegative matrix factorization (NMF) was popularized as a tool for
data mining by Lee and Seung in 1999. NMF attempts to approximate a
matrix with nonnegative entries by a product of two low-rank
matrices, also with nonnegative entries. We propose an algorithm
called rank-one downdate (R1D) for computing a NMF that is partly motivated
by singular value decomposition. This algorithm
computes the dominant singular
values and vectors of adaptively determined submatrices of a matrix.
On each iteration, R1D extracts a rank-one submatrix from the dataset
according to an objective function.
We establish a theoretical result that  maximizing
this objective function corresponds to
correctly classifying articles in a nearly
separable corpus. We also provide computational experiments showing
the success of this method in identifying features in realistic
datasets.
\end{abstract}

\section{Nonnegative Matrix Factorization}
Several problems in information retrieval can be posed as
low-rank matrix approximation.  The seminal paper by
Deerwester et al.\ \cite{lsi} on latent
semantic indexing (LSI) showed that approximating
a term-document matrix describing a corpus of articles
via the SVD led to powerful query and classification techniques.
A drawback of LSI is that the low-rank factors in general will
have both positive and negative entries, and there is no obvious
statistical interpretation of the negative entries.  This led
Lee and Seung \cite{LeeSeung} among others to propose {\em nonnegative
matrix factorization}, that is, approximation of a matrix $A\in\R^{m\times}$
as a product of two factors $WH^T$, where $W\in\R^{m\times k}$,
$H\in\R^{n\times k}$, both have nonnegative
entries, and $k\le \min(m,n)$.
Lee and Seung showed intriguing results
with a corpus of images.  In a related work, Hofmann \cite{hofmann}
showed the application of NMF to text retrieval.
Nonnegative matrix factorization has its roots in work
of Gregory \cite{Gregory}, Paatero \cite{Paatero} and 
Cohen and Rothblum \cite{CohenRothblum}.

Since the problem is NP-hard \cite{vavasis:nmf_nphard}, it is not surprising that
no algorithm is known to solve NMF to optimality.
Heuristic
algorithms proposed for NMF have generally been based on incrementally 
improving the objective $\Vert A-WH^T\Vert$ in some norm using local
moves.
A particularly sophisticated example of local search is due, e.g.,
to Kim and Park \cite{KimPark}.   A drawback of local search is that
it is sensitive to initialization and also is sometimes difficult to
establish convergence.

We propose an NMF method based on greedy rank-one downdating that we
call R1D.  R1D is partly motived by Jordan's algorithm for computing
the SVD, which is described in Section~\ref{sec:algo}.
Unlike local search methods,
greedy methods do not require an initial guess.
In Section~\ref{sec:svd}, we compare our
algorithm to Jordan's SVD algorithm, which is the archetypal
greedy downdating procedure.
Previous
work on greedy downdating algorithms for NMF is the subject of 
Section~\ref{sec:related}.
In Section~\ref{sec:maintheorem}, we present the main theoretical
result of this paper, which states that in a certain model
of text due to Papadimitriou et al.~\cite{Papadimitriou}, optimizing our
objective function means correctly identifying a topic in a text
corpus.  Similarly, optimization of the objective function
corresponds to identifying a feature in a certain model of an
image database, as demonstrated in Section~\ref{sec:imagetheory}.
We then turn to computational experiments: 
in Section~\ref{sec:image}, we present results for R1D on image databases,
and in Section~\ref{sec:text}, we present results on text.

\section{Algorithm and Objective Function}
\label{sec:algo}

Rank-one downdate (R1D) is based on the simple observation that
the leading singular vectors of a nonnegative matrix are
nonnegative. This is a consequence of the Perron-Frobenius
theorem  \cite{GVL}. Based on this observation, it is trivial to compute
rank-one NMF. This idea can be extended to approximate higher
order NMF. Suppose we compute the
rank-one NMF and then subtract it from the original matrix. The
original matrix will no longer be nonnegative, but all negative
entries can be forced to be zero or positive and the procedure can
be repeated.

An improvement on this idea takes only a submatrix of the original
matrix and applies the Perron-Frobenius theorem.  The point is that
taking the whole matrix will in some sense average the features,
whereas a submatrix can pick out particular features.  A second point
of taking a submatrix is that a correctly chosen submatrix may be
very close to having rank one, so the step of forcing the residuals
to being zero will not introduce significant inaccuracy (since they
will already be close to zero).

The outer loop of the R1D algorithm is as follows.
\begin{tabbing}
\setcounter{algstmt}{0}
++\=+++\=++\=++\=++\=\kill
\> function $[W,H]={\tt R1D}(A,k)$ \\
\>Inputs: $A\in\R^{m\times n}$, $k>0$. \\
\>Outputs: $W\in\R^{m\times k}$, $H\in\R^{n\times k}$. \\
\> \algnumber{r0} \> for $\mu=1,\ldots,k$ \\
\> \algnumber{s1} \> \> $[M,N,\u,\v,\sigma]=\mbox{\tt ApproxRankOneSubmatrix}(A);$ \\
\>\algnumber{r15} \> \> $W(M,\mu)=\u(M)$. \\
\>\algnumber{r16} \> \> $H(N,\mu)=\sigma\v(N)$. \\
\>\algnumber{r17} \> \> $A(M,N)=0.$ \\
\>\algnumber{r18} \>  end for
\end{tabbing}
Here, $M$ is a subset of $\{1,\ldots,m\}$,  $N$ is a subset of
$\{1,\ldots,n\}$, $\u\in\R^m$, $\v\in\R^n$ and $\sigma\in\R$, and
$\u,\v$ are both unit vectors.
We follow Matlab subscripting conventions, so that $\u(M)$ denotes the
subvector of $\u$ indexed by $M$.  In the above algorithm,
$\u(\{1,\ldots,m\}-M)=\bz$ and $\v(\{1,\ldots,n\}-N)=\bz$.
The function {\tt ApproxRankOneSubmatrix} selects 
$M,N,\u(M),\v(N),\sigma$
so that $A(M,N)$ (i.e., the submatrix of
$A$ indexed by row set $M$ and column set $N$)  is approximately rank one,
and in particular, is approximately equal to $\u(M)\sigma\v^T(N)$.

This outer loop for NMF may be called ``greedy rank-one downdating'' since
it greedily tries to fill the columns of $W$ and $H$ from left to right
by finding good rank-one submatrices of $A$ and subtracting them from $A$.
The classical greedy rank-one downdating algorithm is Jordan's algorithm for
the SVD, described in Section~\ref{sec:svd}.  Related work on greedy
rank-one downdating for NMF is the topic of Section~\ref{sec:related}.

The subroutine {\tt ApproxRankOneSubmatrix}, presented later in this
section, is a heuristic routine to maximize the following objective
function:
\begin{equation}
f(M,N,\u,\sigma,\v)=\Vert A(M,N)\Vert_F^2-\gamma \Vert A(M,N)-\u(M)\sigma\v(N)^T\Vert_F^2.
\label{eq:objfunc}
\end{equation}
Here, $\gamma$ is a penalty parameter.
The Frobenius norm of an $m\times n$ matrix $B$,
denoted $\Vert B\Vert_F$, is defined
to be $\sqrt{B(1,1)^2+B(1,2)^2+\cdots+B(m,n)^2}$.
The rationale for \eref{eq:objfunc} is as follows:
the first term in \eref{eq:objfunc}
expresses the objective that $A(M,N)$ should be
large, while the second term penalizes departure of $A(M,N)$ from 
being a rank-one matrix. 

Since the optimal $\u,\sigma,\v$ come from
the SVD (once $M,N$ are fixed), the above objective function can
be rewritten just in terms of $M$ and $N$ as
\begin{eqnarray}
f(M,N) &=&\sum_{i=1}^p \sigma_i(A(M,N))^2 - \gamma\sum_{i=2}^p 
\sigma_i(A(M,N))^2 \nonumber \\
&=& \sigma_1(A(M,N))^2-(\gamma-1) \nonumber \\
& & \quad\mbox{}\cdot
(\sigma_2(A(M,N))^2+\cdots+\sigma_p(A(M,N))^2), \label{eq:objfunc2}
\end{eqnarray}
where $p=\min(|M|,|N|)$.
The penalty parameter $\gamma$ should be greater
than 1 so that the presence of low-rank contributions is
penalized rather than rewarded.

We conjecture that maximizing \eref{eq:objfunc} is NP-hard
(see Section~\ref{sec:nphard}), so we instead propose a heuristic
routine for optimizing it.
The
procedure alternates improving $(\v,N)$ and  $(\u,M)$.
The rationale for this alternation is that for
fixed $(\v,N)$,
the objective function \eref{eq:objfunc} is separable by rows of the
matrix.  Similarly, for fixed $(\u,M)$,
the objective function is separable by columns.  
Let us state and prove this as a lemma.

\begin{lemma}
Let $(\v,N)$ be the optimizing choice of these variables
in \eref{eq:objfunc}.  Then the optimal $M$ is determined
as follows.   Define 
\begin{equation}
g_i= -A(i,N)A(i,N)^T+\bar\gamma(A(i,N)\v(N))^2,
\label{eq:gidef}
\end{equation}
where 
\begin{equation}
\bar\gamma=\gamma/(\gamma -1).
\label{eq:bargammadef}
\end{equation}
Then $i\in M$ if $g_i\ge 0$.  (If exact equality $g_i=0$ holds, 
then including $i$ or not does not affect optimality.)
Furthermore, $\sigma u_i$ is optimally chosen to be $A(i,N)\v(N)$.
\label{lem:separable}
\end{lemma}

\noindent {\bf Remark 1.}  The lemma gives the formula
for optimal $\sigma u_i$ for each $i$, i.e., the  formula for 
the optimal $\sigma\u(M)$.
To obtain a formula for optimal $\u$ and $\sigma$ separately, we define
$\sigma:=\Vert \sigma\u(M)\Vert$ 
and $\u(M):=\sigma\u(M)/\Vert\sigma\u(M)\Vert$.

\noindent {\bf Remark 2.} Assuming instead that the optimizing choice
$(\u,M)$ is given, there is a similar formula for determining
membership in $N$.  Define
\begin{equation}
f_j= -A(M,j)^TA(M,j)+\bar\gamma(A(M,j)^T\u(M))^2,
\label{eq:fjdef}
\end{equation}
and take $N=\{j:f_j\ge 0\}$.

\begin{proof}
Observe that
$$f(M,N,\u,\sigma,\v)=\sum_{i=1}^m\chi_M(i)
\left(\Vert A(i,N)\Vert^2 - \gamma \Vert A(i,N)-\beta_i\v(N)^T\Vert^2\right),$$
where $\beta_i=\sigma u_i$ and $\chi_M(i)=1$ for $i\in M$ and
$\chi_M(i)=0$ for $i\notin M$.  Observe that $\beta_i$ occurs
only in the $i$th term of the above summation, hence assuming
$\v$ and $N$ are optimal, each
term may be optimized separately.  The optimal $\beta_i$ 
(that is, the minimizer of  $\Vert A(i,N)-\beta_i\v(N)^T\Vert$)
is $A(i,N)\v(N)$, the solution
to a simple linear least-squares minimization.
Thus, we conclude that putting row $i$ into index set
$M$ is improves the objective function if and only if
$g_i\ge 0$, where
$$g_i=\Vert A(i,N)\Vert^2 - \gamma \Vert A(i,N)-A(i,N)\v(N)\v(N)^T\Vert^2.$$
The formula for $g_i$ can be simplified as follows:
\begin{eqnarray*}
g_i
&=& A(i,N)A(i,N)^T \\
& & \quad\mbox{} - \gamma(A(i,N)-A(i,N)\v(N)\v(N)^T)(A(i,N)-A(i,N)\v(N)\v(N)^T)^T \\
&=& -(\gamma-1)A(i,N)A(i,N)^T+\gamma(A(i,N)\v(N))^2.
\end{eqnarray*}
Rescaling by $\gamma-1$ (which does not affect the acceptance
criterion) and substituting \eref{eq:bargammadef}, we
that row $i$ makes a positive contribution to the objective function
provided
$\bar\gamma(A(i,N)\v(N))^2-A(i,N)A(i,N)^T>0.$
\end{proof}

The next issue is choice of starting guess for $M,N,\u,\v,\sigma$.
The algorithm should be initialized with a starting guess that
has a positive score, else the rules for discarding rows and columns
could conceivable discard all rows or columns.  More strongly, in order
to improve the score of converged solution, it seems sensible to
select a starting guess with a high score.  For this reason, 
R1D uses as its starting guess a single column of $A$, and in particular,
the column of $A$ with the greatest
norm.  (A single row may also be chosen.)
It then chooses $\u$ to be the normalization of this
column.
This column is exactly rank one, so for the correct values of
$\sigma$ and $\v$ the first penalty term of 
\eref{eq:objfunc} is zero.  
We have derived the following algorithm
for the subroutine
{\tt ApproxRankOneSubmatrix} occurring in statement
\stmtref{s1} in R1D.

\begin{tabbing}
\setcounter{algstmt}{0}
++\=+++\=++\=++\=++\=\kill
\> function $[M,N,\u,\v,\sigma]=\mbox{\tt ApproxRankOneSubmatrix}(A);$ \\
\>Input: $A\in\R^{m\times n}$. \\
\>Outputs: $M\subset\{1,\ldots,m\}$, $N\subset\{1,\ldots,n\}$, 
$\u\in\R^m$, $\v\in\R^n$, $\sigma\in\R$. \\
\>Parameter: $\bar\gamma> 1$ \\
\> \algnumber{r1} \> Select $j_0\in\{1,\ldots,n\}$ to maximize 
$\Vert A(:,j_0)\Vert$. \\
\> \algnumber{r2} \> $M =\{1,\ldots,m\}$. \\
\> \algnumber{r3} \> $N=\{j_0\}$. \\
\> \algnumber{r4} \> $\sigma = \Vert A(:,j_0)\Vert.$\\
\> \algnumber{r5} \> $\u=A(:,j_0)/\sigma$. \\
\> \algnumber{r7} \> Repeat \\
\> \algnumber{r8} \> \> $\bar\v =A(M,:)^T\u(M).$ \\
\> \algnumber{r9} \> \> $N = \{j: \bar\gamma\bar v(j)^2-\Vert A(M,j)\Vert^2
>0\}$. \\
\> \algnumber{r10} \> \> $\v(N) = \bar\v(N)/\Vert\bar\v(N)\Vert.$ \hspace{0.5in}  /* Other entries of $\v$ unused */ \\
\> \algnumber{r11} \> \> $\bar \u=A(:,N)\v(N)$. \\
\>\algnumber{r12} \> \> $M=\{i: \bar\gamma\bar u(i)^2-\Vert A(i,N)\Vert^2
>0\}$. \\
\>\algnumber{r13} \>\> $\sigma = \Vert\u(M)\Vert$.\\
\>\algnumber{r13a}  \> \>$\u(M)=\bar\u(M)/\sigma$. \hspace{0.5in} /* Other entries of $\u$ unused */\\
\>\algnumber{r14}  \> until stagnation in $M,N,\u,\v,\sigma$.
\end{tabbing}

The `Repeat' loop is guaranteed to make progress because each iteration
increases the value of the objective function.  On the other hand, there
does not seem to be any easy way to derive a useful prior upper bound on
its number of iterations.  In practice, it proceeds quite quickly, usually
converging in 10--15 iterations.  But to guarantee fast termination, 
monotonicity can be forced on $M$ and $N$ by requiring $M$ to shrink and $N$ to
grow.  In other words, statement \stmtref{r9} can be replaced by
$$N = N\cup \{j: \bar\gamma\bar v(j)^2-\Vert A(M,j)\Vert^2>0\},$$
and statement \stmtref{r12} by
$$M=M-\{i: \bar\gamma\bar u(i)^2-\Vert A(i,N)\Vert^2\le 0\}.$$
Our experiments indicate that this change does not have a major impact
on the performance of R1D.

Another possible enhancement to the algorithm is as follows: we modify the
objective function by adding a second penalty term 
\begin{equation}
-\rho|M|\cdot |N|
\label{eq:secondpenalty}
\end{equation}
to \eref{eq:objfunc}
where $\rho>0$ is a parameter.
The purpose of this term is to penalize very low-norm rows
or columns from being inserted into $A(M,N)$ since they are probably
noisy.  For data with larger norm, the first term of
\eref{eq:objfunc} should dominate this penalty.  Notice that this
penalty term is also separable so it is easy to implement:
the formula in \stmtref{r9} is changed
to $\bar\gamma\bar v(j)^2-\Vert A(M,j)\Vert^2-\bar \rho|M|>0$ while the formula
in \stmtref{r12} becomes
$\bar\gamma\bar u(i)^2-\Vert A(i,N)\Vert^2-\bar\rho|N|>0$, where
$\bar \rho=\rho/(\gamma-1)$.  We may select
$\bar \rho$ so that 
the third term is a small fraction
(say $\bar\eta=1/20$) of the other terms
in the initial starting point.
This leads to the following
definition for $\rho$:
$$\rho=\bar\eta(\bar\gamma-1)\sigma^2 / m,$$
which may be computed immediately after \stmtref{r4}.

\section{Relationship to the SVD}
\label{sec:svd}
The classical rank-one greedy downdating algorithm is
Jordan's algorithm for computing the singular value
decomposition (SVD) \cite{Stew93}.
Recall that the SVD takes as input an $m\times n$ matrix $A$
and returns three factors $U,\Sigma,V$ such that 
$U\in\R^{m\times k}$ and $U$ has orthonormal columns
(i.e., $U^TU=I$), $\Sigma\in\R^{k\times k}$ and is diagonal
with nonnegative diagonal entries, and $V\in\R^{n\times k}$
also with orthonormal columns, such that $U\Sigma V^T$ is
the optimal rank-$k$ approximation to $A$ in either the 2-norm
or Frobenius norm.  (Recall that the 2-norm of an $m\times n$
matrix $B$, denoted $\Vert B\Vert_2$, is
defined to be $\sqrt{\lambda_{\max}(B^TB)}$, where $\lambda_{\max}$
denotes the maximum eigenvalue.)

\begin{tabbing}
\setcounter{algstmt}{0}
++\=+++\=++\=++\=++\=\kill
\>$[U,\Sigma,V]={\tt JordanSVD}(A,k);$ \\
\>Input: $A\in\R^{m\times n}$ and $k\le \min(m,n)$. \\
\>Outputs: $U,\Sigma,V$ as above. \\
\>\algnumber{svd1}\> for $\mu=1,\ldots,k$ \\
\>\algnumber{svd2}\> \> Select a random nonzero $\bar\u\in\R^m$. \\
\> \algnumber{svd4} \> \>$\sigma = \Vert \bar \u\Vert.$\\
\> \algnumber{svd5} \> \>$\u=\bar \u/\sigma$. \\
\> \algnumber{svd7} \> \>Repeat \hspace{0.5in} /* power method */\\
\> \algnumber{svd8} \> \>\> $\bar\v =A^T\u.$ \\
\> \algnumber{svd10} \>\> \> $\v = \bar\v/\Vert\bar\v\Vert.$  \\
\> \algnumber{svd11} \>\> \> $\bar \u=A\v$. \\
\>\algnumber{svd13} \>\> \>$\sigma = \Vert\bar\u\Vert$.\\
\>\algnumber{svd13a}  \> \>\>$\u=\bar\u/\sigma$.\\
\>\algnumber{svd14}  \> \>until stagnation in $\u,\sigma,\v$. \\
\>\algnumber{svd15}  \> \> $A=A-\u\sigma\v^T;$ \\
\>\algnumber{svd16}  \> \> $U(:,\mu)=\u;$ \\
\>\algnumber{svd17}  \> \> $V(:,\mu)=\v;$ \\
\>\algnumber{svd18}  \> \> $\Sigma(\mu,\mu)=\sigma;$ \\
\>\algnumber{svd19}\> end for
\end{tabbing}
Thus, we see that R1D is quite similar to the SVD.  The principal
difference is that R1D tries to find a submatrix indexed
by $M\times N$ at
the same time that it tries to identify the optimal $\u$ and $\v$.
Because of this similarity, the formulas for $\u$ and $\v$ occurring in 
\stmtref{r10} and \stmtref{r13a} of subroutine {\tt ApproxRankOneSubmatrix},
which were presented earlier as solutions to a least-squares problem,
may also be regarded as  steps in a power method.  In fact, 
if $M$ and $N$ are fixed, then the inner Repeat-loop 
of that subroutine will indeed converge to the dominant singular
triple of $A(M,N)$.

As noted earlier, use of the SVD
on term-document matrices dates back to latent semantic indexing due
to Deerwester et al.~\cite{lsi}.  Its effectiveness 
at creating a faithful low-dimensional model of a corpus in the
case of separable corpora was established by Papadimitriou et 
al.~\cite{Papadimitriou}.  Although not originally
 proposed specifically as a clustering
tool, the SVD has been observed to find good clusters in some
settings \cite{DhillonModha}.

The SVD, however, has a significant shortcoming as far as its use
for clustering.
Consider the following term-document matrix $A$,
which is a sum of a completely separable matrix $B$ and
noise matrix $E$:
\begin{eqnarray*}
A&=& B+ E \\
&=&
\left(
\begin{array}{cccc}
1.01 & 1.01 & 0 & 0 \\
1.01 & 1.01 & 0 & 0 \\
0 & 0 & 1 & 1 \\
0 & 0 & 1 & 1
\end{array}
\right)
+
\left(
\begin{array}{cccc}
-0.02 & -0.02 & 0.02 & 0.02 \\
0& 0 & 0 & 0 \\
0 & 0 & 0 & 0 \\
0 & 0 & 0 & 0
\end{array}
\right).
\end{eqnarray*}
It should be clear that there are two separate topics in $A$ given
by the two diagonal blocks, and a reasonable NMF algorithm ought
to be able to identify the two blocks.  In other words, for $k=2$,
one would expect an answer close to
$$W=H=\left(
\begin{array}{cc}
1 & 0 \\
1 & 0 \\
0 & 1 \\
0 & 1
\end{array}
\right).$$
Perhaps unexpectedly, the dominant right singular vector
of $A$ is very close to being proportional to $[1;1;1;1]$, i.e.,
the two topics are entangled in one singular vector.  The reason
for this behavior is that the matrix $B$ has two nearly equal
singular values, so its singular vectors are highly sensitive
to small perturbations (such as the matrix $E$). 
R1D avoids
this pitfall by computing the dominant singular vector of a submatrix
of the original $A$ instead of the whole matrix.

\section{Related Work}
\label{sec:related}
As mentioned in the introduction, most algorithms proposed in the
literature are based on forming an initial $W$ and $H$ and then
improving them by local search on an objective function.
The objective function usually includes a term of the form
$\Vert A-WH^T\Vert$ in some norm, and may include other terms.

A few previous works follow an approach similar to ours, namely,
greedy subtraction of rank-one matrices.  This includes
the work of Bergmann et al. \cite{bergman}, who identify the rank-one
matrix to subtract as the fixed point of an iterative process.
Asgarian and Greiner \cite{greiner} find the dominant singular
pair and then truncate it.  Gillis \cite{gillis} finds a
rank-one underestimator and subtracts that.    Boutsidis and Gallopoulos
\cite{Boutsidis}
consider the use of a greedy algorithm for initializing other
algorithm and make the following interesting observation:
The nonnegative part of a rank-one matrix has rank at most 2.

The main innovation herein is the idea that the search for the
rank-one submatrix should itself be an optimization subproblem.
This observation allows us to compare and rank one 
candidate submatrix to
another.
(Gillis also phrases his subproblem as optimization, although his
optimization problem does not explicitly seek submatrices like ours.)
A second innovation is our analysis showing in Section~\ref{sec:maintheorem}
that if the subproblemn
were solved optimally, then R1D would be able to accurately find
the topics in the Papadimitriou et al.~\cite{Papadimitriou} model of
$\eps$-separable corpora.

\section{Behavior of this objective function on a nearly separable corpus}
\label{sec:maintheorem}

In this section, we establish the main
theoretical result of the paper, namely, that the objective function
given by \eref{eq:objfunc} is able to correctly identify a topic
in a nearly separable corpus.  We define our {\em text model}
as follows. There is a universe of {\em terms} numbered
$1,\ldots,m$.  There is also a set of {\em topics}
numbered $1,\ldots,t$.  Topic $k$, for $k=1,\ldots,t$, is a probability
distribution over the terms.  Let $P(i,k)$ denote the probability of
term $i$ occurring in topic $k$.  Thus, $P$ is a singly stochastic
matrix, i.e., it has nonnegative entries with column sums exactly 1.
We assume also that there is a probability
distribution over topics; say the probability of topic $k$ is $\tau_k$,
for $k=1,\ldots,t$.  The text model is thus specified
by $P$ and $\tau_1,\ldots,\tau_t$.
We use 
the Zipf distribution
as the model of document length.  In particular, there is a number $L$ such
that all documents have length less than $L$,  and the probability that
a document of length $l$ occurs is
$$\frac{1/l}{1+1/2+\cdots+1/(L-1)}.$$
We have checked that the Zipf model is a good fit for several common
datasets. {\bf [SHOW SOME DATA.]}

A {\em document} is generated from this text model as follows.  First, topic
$k$ is chosen at random according to the probability distribution 
$\{\tau_1,\ldots,\tau_t\}$.  Then, a length $l$ is chosen at random 
from $\{1,\ldots,L-1\}$ according to the
Zipf distribution.  Finally, the document itself is chosen at random
by selecting $l$ terms independently
according to the probability distribution $P(:,k)$.
A {\em corpus} is
a set of $n$ documents chosen independently using this text model.
Its {\em term-document matrix} is the $m\times n$ matrix $A$ such that
$A(i,j)$ is the frequency of term $i$ in document $j$.

We further assume that the text model  is {\em $\eps$-separable}, 
meaning that
each topic $k$ is associated with a set of terms $T_k\subset \{1,\ldots,m\}$,
that $T_1,\ldots,T_t$ are mutually disjoint, and that 
$P(i,k)\le \eps$ for $i\notin T_k$, i.e., 
the probability that a document on topic $k$ will
use a term outside of $T_k$ is small.   Parameter $\epsilon$
must satisfy some inequalities described below.
This corpus model is quite similar to the model of Papadimitriou
et al.\ \cite{Papadimitriou}.  One difference is in the the document
length model.  Our model also relaxes several assumptions
of Papadimitriou et al.

Our main theorem is that the objective function in the previous
section can correctly find documents associated with a particular
topic in a corpus.

\begin{theorem}
Let $(P,(\tau_1,\ldots,\tau_t))$ specify a text model,
and let $\alpha>0$ be chosen arbitrarily.
Suppose there exists an $\eps\ge 0$ satisfying \eref{eq:epsdef} below
such that the text-model is $\eps$-separable with respect
to $T_1,\ldots, T_t$, the subsets of terms defining the topics.
Let $A$ be the term-document matrix of a corpus 
of $n$ documents drawn from this model when the 
document-length parameter is $L$.
Choose $\gamma=4$
in \eref{eq:objfunc}.
Then with probability tending to $1$ as $n\rightarrow\infty$
and $L\rightarrow\infty$
(refer to Assumption {\rm A1} below),
the optimizing pair $(M,N)$ of \eref{eq:objfunc} satisfies the following.
Let $D_1,\ldots,D_t$ be the partitioning of the columns of $A$
according to topics.
There exists a topic $k\in \{1,\ldots,t\}$ such that
$A(M,N)$ and $A(T_k,D_k)$ are nearly coincident in the following
sense.
$$
\sum_{(i,j)\in (M\times N) \bigtriangleup (T_k\times D_k)}
A(i,j)^2\le \alpha
\sum_{(i,j)\in M\times N} A(i,j)^2.$$
\label{thm:main}
\end{theorem}
Here, $X\bigtriangleup Y$ denotes the set-theoretic symmetric
difference $(X-Y)\cup(Y-X)$.

The organization of the proof of this theorem is as follows.
We first analyze the Zipf distribution and propose some assumptions
that hold with probability tending to $1$ as $n,L\rightarrow \infty$.
Under these assumptions, we make some preliminary estimates of
norms of submatrices of $A$.  Then we prove a sequence of
lemmas as follows.
\begin{itemize}
\item
In Lemma~\ref{lem:foptlb}, we establish a lower bound on the
optimal value
of the objective function by analyzing the objective function
value with the choice $M=T_k$, $N=D_k'$, where $D_k'\subset D_k$
are the `acceptable' documents from $D_k$ (defined below).
\item
In Lemma~\ref{lem:heavynotacc}, we establish an upper bound
on the contribution from unacceptable entries to the optimal solution.
\item
In Lemma~\ref{lem:heavyacc}, we deduce as a consequence of 
the two preceding lemmas that heavy
acceptable entries must compose a significant portion of the
optimal solution.  Here, an entry $A(i,j)$ is {\em heavy}
if $P(i,k)\ge\chi$,
where $\chi$is a scalar defined below and $k$ is the topic of document $j$.
\item
In Lemma~\ref{lem:uniqueheavy}, we show that the optimal solution
cannot contain heavy acceptable entries from two different topics.
\item
Thus, the preceding lemmas imply that heavy acceptable entries from
a single topic $k$ must dominate the optimal solution.  Therefore,
we show in Lemma~\ref{lem:approxsingvec} that the
left and right  singular vectors
of the optimal $A(M,N)$ can be estimated from $P(M,k)$ and
the vector of lengths of documents indexed by $N$ respectively.
\item
In Lemma~\ref{lem:incrobj}, we give a general condition under
which adding a row or column to $M$ or $N$ could improve the objective
function value. 
\item
In Lemma~\ref{lem:dksubsN}, we show that any column in $D_k'$ satisfies
the condition given by Lemma~\ref{lem:incrobj} (because of the estimate
of the left singular vector given by Lemma~\ref{lem:approxsingvec}), and
therefore $D_k'\subset N$ if $N$ is optimal.
\item
In Lemma~\ref{lem:hksubsM}, we establish using analogous reasoning that
the heavy terms $H_k$ of topic $k$ must be a subset of $M$ if $M$ is optimal.
\item
Finally, the theorem can be proved because all entries of $A(M,N)$
that are not from $H_k\times D_k'$ are either not heavy or unacceptable
but in either case, must have small norm.
\end{itemize}

We start by stating the inequality that $\eps$ must satisfy in order
for the theorem to hold.  It should be noted that the constants
that follow are quite large but are likely large overestimates.
Let $P_{\min}=\min\{P(i,k): i\in T_k, k=1,\ldots,t\}$.  Without loss
of generality, $P_{\min}>0$ since any row $i\in T_k$ such that $P(i,k)=0$
may be removed from $T_k$ without affecting the validity of the model.
The proof requires four parameters, $\eps$, $\theta$, $\phi$ and
$\chi$.  The parameters depend on $\alpha$,
$m$, $t$,  and $P_{\min}$.
They do not depend on 
$n$ and $L$ (since the theorem requires $n\rightarrow\infty$
and $L\rightarrow\infty$).

First, we define 
$$\chi=\min(\chi_1,\chi_2),$$
where 
\begin{eqnarray}
\chi_1&=&\sqrt{3/(32\cdot 16\cdot 25^2\cdot 256t)}/m, 
\label{eq:chidef1} \\
\chi_2&=&\sqrt{3\alpha/(32\cdot512t)}/m.
\label{eq:chidef2}
\end{eqnarray}
Next, we choose
$$\theta=\min(\theta_1,\theta_2,\theta_3),$$
where
\begin{eqnarray}
\theta_1&=&P_{\min}/2,
\label{eq:thetadef1} \\
\theta_2 &=& \sqrt{3/(8\cdot 18\cdot 278\cdot t)}/m,
\label{eq:thetadef5} \\
\theta_3 &=&\chi\cdot\sqrt{1/(6\cdot278\cdot mt)}.
\label{eq:thetadef6}
\end{eqnarray}
Third, take
$$\phi=\min(\phi_1,\phi_2,\phi_3),$$
where
\begin{eqnarray}
\phi_1 &=& 3/(16\cdot256\cdot3\cdot2\cdot25^2mt),
\label{eq:phidef2} \\
\phi_2 &=& 3\chi^2/(64\cdot 6\cdot278\cdot mt),
\label{eq:phidef3} \\
\phi_3 &=& \alpha/(32\cdot 512mt).
\label{eq:phidef4}
\end{eqnarray}
Last,
\begin{equation}
\eps=\min(\sqrt{3/(10m)},\chi,\theta).
\label{eq:epsdef}
\end{equation}

We start with our assumptions of the form 
that $n$ or $L$ must be sufficiently large.  The inequalities
in this assumption are needed below.
Let $n_k=|D_k|$, that is, the number of documents on topic $k$,
$k=1,\ldots,t$. Let
$\tau_{\min}=\min(\tau_1,\ldots,\tau_t)$.

\vspace{0.05in}
{\bf Assumption A1.}
Let $q=\log_2 L$.
Assume $n$ and $L$ are sufficiently large so that all of the following
are valid.  
\begin{eqnarray*}
m\exp(-2L^{1/2}\theta^2)&\le&\phi, \\
L &\ge & 3q/(16\phi), \\
n &\ge & q^3, \\
n\tau_{\min} & \ge & 20q. \\
\end{eqnarray*}
The first two inequalities are lower bounds on $L$, and the last two
say that $L$ cannot grow much faster than $n$.
Also, assume $L$ is a power of 4 so that $q$ is an even
integer.  (This last assumptions is not
necessary but simplifies notation.)

\vspace{0.05in}
The next four assumptions are also needed
for our analysis and are valid with probability
tending to 1 as $n,L\rightarrow\infty$ provided Assumption
A1 holds.
The mean value for $n_k$ is $n\tau_k$, so let us impose
the following assumption.

\vspace{0.05in}
\noindent{\bf Assumption A2.}
For each $k$, $n_k\tau_k/2 \le n_k \le 2n_k\tau_k$.

\vspace{0.05in}
By the Chernoff-Hoeffding bound and the union bound, this assumption
will fail with probability at most
$t\exp(-2n\tau_{\min})$.  This quantity tends to 0 as $n\rightarrow\infty$.

Next, let us provide some estimates for the Zipf distribution.  
Let us partition $D_k$ into subsets $C_{k,1},\ldots,C_{k,q}$ where 
$C_{k,\iota}$ contains
documents of $D_k$ of
length $[2^{\iota-1}, 2^\iota)$ and $q=\log_2 L$, an even integer
by assumption.
It follows from underestimating the Zipf distribution using an integral
that the probability that a document lies in $C_{k,\iota}$ is at least
$(n_k/q-2)/n_k=1/q-2/n_k$.  
We have assumed in A2 that $n_k\ge n\tau_{\min}/2$ and in A1 that
$n\tau_{\min}/2\ge 10q$.
The choice of lengths are independent
trials, and the mean size of $C_{k,\iota}$ is at least $n_k/q-2$,
All of these bounds lead to the following.

\vspace{0.05in}
\noindent{\bf Assumption A3.}
For each $k=1,\ldots,t$ and $\iota=1,\ldots,q$,
$|C_{k,\iota}|\ge n_k/(2q)$.

\vspace{0.05in}
The probability
of failure of this assumption is at most $qt\exp(-n_k/(8q))$
which again tends to zero since $n_k/(8q)\ge n\tau_{\min}/(16q)$
by A2 and $n\tau_{\min}/(16q)\ge
q^2\tau_{min}/16$ by A1, and finally, $q\rightarrow\infty$.
A consequence of A3 is that
the number of documents that have
length at least $L^{1/2}$ (i.e., those in $C_{k,q/2+1}\cup\cdots\cup C_{k,q}$)
is at least $n_k/4$.
We also need an upper bound on $|C_{k,\iota}|$.  The mean value
of this quantity is at most $n_k/q+1$ using an integral to overestimate
the Zipf distribution.  Since the documents are chosen independently,
we obtain the following.

\vspace{0.05in}
\noindent{\bf Assumption A4.}
For each
$k=1,\ldots,t$ and $\iota=1,\ldots,q$,
$|C_{k,\iota}|\le 2n_k/q$.

\vspace{0.05in}
The probability of failure is
$qt\exp(-2n_k/q^2)$ by the Chernoff-Hoeffding bound. Using arguments 
similar to those in the previous paragraph, this tends to 0
as $n,L\rightarrow\infty$ under A1.

Let $j\in D_k$ index a document on topic $k$ whose length we denote
as $l_j$.
The mean value for $A(:,j)$ is $l_jP(:,k)$ by
the properties of the multinomial distribution.  Let us now consider
the probability that any $A(i,j)$ diverges significantly 
from the mean, e.g., say $|A(i,j)-l_jP(i,k)|\ge l_j\theta$.
Again, by the Chernoff-Hoeffding bound, this probability is at
most $\exp(-2l_j\theta^2)$, so using a union bound, the probability
that any entry will diverge by $l_j\theta$ from its mean
is at most $m\exp(-2l_j\theta^2)$.  If we further assume $l_j\ge L^{1/2}$, 
this quantity
is at most $m\exp(-2L^{1/2}\theta^2)$.   
We have assumed in A1 that  $L$ is sufficiently
large so that $m\exp(-2L^{1/2}\theta^2)\le \phi$, where
$\phi$ is the parameter
given by \eref{eq:phidef2}--\eref{eq:phidef4}. 

We say that a column $j\in D_k$
is {\em acceptable} if its length $l_j$ is at least $L^{1/2}$
and if the distance of each entry from its mean is
at most $\theta l_j$.  Let $D_{\rm acc}$ denote the subset of 
$\{1,\ldots,n\}$ of acceptable documents and $D_{\rm unacc}$ its
complement.
By the assumptions so far, the number of documents 
with length at least $L^{1/2}$ in topic $k$ is at least $n_k/4$.
Let $D_k'$ denote $D_k\cap D_{\rm acc}$, i.e., 
the acceptable subset of $D_k$ and
let $C_{k,{q/2+1}}',\ldots,C_{k,q}'$ denote the acceptable subsets of
$C_{k,q/2+1},\ldots,C_{k,q}$. 
We now impose the last assumption.

\vspace{0.05in}
\noindent
{\bf Assumption A5.}
The acceptable subset of each $C_{k,\iota}$ ($\iota=q/2+1,\ldots,q$)
has size at least $|C_{k,\iota}|(1-2\phi)$.  

\vspace{0.05in}
By the union bound,
this assumption fails with probability at most
$$\sum_{k=1}^t\sum_{\iota=q/2+1}^q \exp(-2|C_{k,\iota}|\phi),$$
which, according to prior assumptions, is at most
$(qt/2)\exp(-n\phi\tau_{\min}/(2q))$. 
Again, from A1 and A2, this probability tends to 0 for large
$n$ and $L$.

Let us now 
derive some inequalities useful for the upcoming analysis.
A simple inequality is
\begin{equation}
\Vert A(:,j)\Vert \le l_j
\label{eq:ajltlj}
\end{equation}
which follows because of the inequality
$\Vert\x\Vert_2\le\Vert\x\Vert_1$.
Another simple inequality is that if $a,b$ are both nonnegative,
then 
\begin{equation}
(a-b)^2\le a^2+b^2.
\label{eq:aminusb}
\end{equation}

Let $\l\in\R^n$ denote the vector $(l_1,\ldots,l_n)$
of document lengths. 
We now establish some needed norm estimates for $\l$.
\begin{eqnarray}
\Vert \l(D_k')\Vert^2 &=&
\sum_{j\in  D_k'}l_j^2 
\nonumber \\
& = &
\sum_{\iota=q/2+1}^q\sum_{j\in C_{k,\iota}'}l_j^2 
\nonumber \\
&\ge & 
\sum_{\iota=q/2+1}^q\sum_{j\in C_{k,\iota}'} 2^{2\iota-2}
\nonumber \\
&=&
\sum_{\iota=q/2+1}^q 2^{2\iota-2} |C_{k,\iota}'|
\nonumber \\
&\ge &
\sum_{\iota=q/2+1}^q 2^{2\iota-2} (1-2\phi)|C_{k,\iota}|
\nonumber \\
&\ge &
\sum_{\iota=q/2+1}^q 2^{2\iota-2} (1-2\phi)n_k/(2q)
\nonumber \\
&\ge &
(1-2\phi)n_kL^2/(8q).
\label{eq:ldkplb}
\end{eqnarray}
Here, Assumption A5 was used for the fifth line and A3 for second.
A useful upper bound is:
\begin{eqnarray}
\Vert \l(D_k)\Vert^2 &=& 
\sum_{j\in  D_k}l_j^2 
\nonumber \\
& = &
\sum_{\iota=1}^q\sum_{j\in C_{k,\iota}}l_j^2 
\nonumber \\
&\le & 
\sum_{\iota=1}^q\sum_{j\in C_{k,\iota}}2^{2\iota}
\nonumber \\
&=&
\sum_{\iota=1}^q 2^{2\iota} |C_{k,\iota}|
\nonumber \\
&\le &
\sum_{\iota=1}^q 2^{2\iota} \cdot 2n_k/q
\nonumber \\
&\le &
8n_kL^2/(3q).
\label{eq:ldkub}
\end{eqnarray}
Since $\Vert\l\Vert^2=\Vert\l(D_1)\Vert^2+\cdots+\Vert\l(D_t)\Vert^2$,
\begin{equation}
\Vert\l\Vert^2 \le 8nL^2/(3q).
\label{eq:lub}
\end{equation}

Some final estimates concern the sum of
squares of lengths of unacceptable documents.  A document can be
unacceptable either because its length is less than $L^{1/2}$
(i.e., it lies in $C_{k,\iota}$ for some $k=1,\ldots,t$ and some
$\iota=1,\ldots,q/2$) or else because its term frequencies deviate too
much from the mean (by more than $l_j\theta$ in some position). 
In the former case, all document lengths are bounded by $L^{1/2}$,
hence squared document lengths are bounded by $L$.
For the latter case, we can apply Assumption A5.
Thus, we have the following estimate on unacceptable
documents:
\begin{eqnarray*}
\Vert\l(D_k-D_k')\Vert^2 &=&
\Vert\l(C_{k,1}\cup\cdots\cup C_{k,q/2})\Vert^2 \\
& & \quad\mbox{}+
\Vert\l(C_{k,q/2+1}\cup\cdots\cup C_{k,q}-C_{k,q/2+1}'-\cdots
-C_{k,q}')\Vert^2 \\
&\le & n_kL + \sum_{\iota=q/2+1}^q |C_{k,\iota}-C_{k,\iota}'| 2^{2\iota} \\
&\le & n_kL + \sum_{\iota=q/2+1}^q 2\phi|C_{k,\iota}| 2^{2\iota} \\
&\le & n_kL+16\phi n_k L^2 / (3q).
\end{eqnarray*}
Here, Assumption A3 was used for the fourth line.
We can combine these contributions from individual topics to obtain
the upper bound:
\begin{equation}
\Vert \l(D_{\rm unacc})\Vert^2 \le
nL+16\phi nL^2/(3q).
\label{eq:ldunacc0}
\end{equation}

These estimates can be extended to sum of squares of the entries of
$A$:
\begin{eqnarray*}
\Vert A(:,D_k-D_k')\Vert_F^2 &= &
\sum_{j\in D_k-D_k'}\sum_{i=1}^m A(i,j)^2 \\
& \le &
\sum_{j\in D_k-D_k'}l_j^2 \\
&\le &
n_kL+16\phi n_kL^2/(3q).
\end{eqnarray*}
The second line follows from \eref{eq:ajltlj}.

Thus,
$$\Vert A(:,D_{\rm unacc})\Vert^2_F \le nL+16\phi nL^2/(3q).$$
Recalling from Assumption A1 that  $L\ge 3q/(16\phi)$,
the second term dominates in the above four inequalities, so
\begin{eqnarray}
\Vert \l(D_k-D_k')\Vert^2 &\le& 32\phi n_kL^2/(3q), \nonumber\\
\Vert A(:,D_k-D_k')\Vert_F^2 &\le&32\phi n_kL^2/(3q), \nonumber\\
\Vert \l(D_{\rm unacc})\Vert^2 &\le& 32\phi nL^2/(3q), \label{eq:ldunacc1} \\
\Vert A(:,D_{\rm unacc})\Vert_F^2 &\le& 32\phi nL^2/(3q).
\label{eq:adunacc0}
\end{eqnarray}
Because of \eref{eq:phidef2},
\begin{eqnarray}
\Vert \l(D_k-D_k')\Vert^2 &\le& n_kL^2/(16\cdot256\cdot25^2\cdot qtm), \\
\Vert A(:,D_k-D_k')\Vert_F^2 &\le& n_kL^2/(16\cdot256\cdot25^2\cdot qtm),\\
\Vert \l(D_{\rm unacc})\Vert^2 &\le& nL^2/(16\cdot256\cdot25^2\cdot qtm), 
\label{eq:ldunacc}\\
\Vert A(:,D_{\rm unacc})\Vert_F^2 &\le& nL^2/(16\cdot256\cdot25^2\cdot qtm).
\label{eq:adunacc}
\end{eqnarray}

With these preliminary inequalities in hand, we may
now begin the first lemma in the proof of the main theorem.

\begin{lemma}
Under Assumptions {\rm A1--A5},
\begin{equation}
\fopt \ge nL^2/(256qtm),
\label{eq:fopt}
\end{equation}
where $\fopt$ denotes the optimal value of \eref{eq:objfunc}.
\label{lem:foptlb}
\end{lemma}
\begin{proof}
The proof follows from
estimating the value of the objective function for the
choices
$M=T_k$ and $N=D_k'$.  
We can estimate the first term in \eref{eq:objfunc} as
\begin{eqnarray}
\Vert A(M,N) \Vert_F^2 &=& \sum_{j\in D_k'}\sum_{i_\in T_k} A(i,j)^2 
\nonumber \\
&\ge &\sum_{j\in D_k'}\sum_{i\in T_k} l_j^2(P(i,k)-\theta)^2 
\nonumber \\
& \ge &\sum_{j\in D_k'}\sum_{i\in T_k} l_j^2P(i,k)^2/4 
\nonumber \\
& = & 
\Vert \l(D_k')\Vert^2\cdot \sum_{i\in T_k}P(i,k)^2/4 
\nonumber \\
& \ge & (1-2\phi)n_kL^2\Vert P(T_k,k)\Vert^2/(32q) 
\nonumber \\
& \ge & (1-2\phi)n_kL^2/(64qm) 
\nonumber \\
&\ge & n_kL^2/(128qm).
\label{eq:fopt1} 
\end{eqnarray}
The second line follows by the definition of `acceptable.'  The third follow
because $\theta\le P(i,k)/2$ by \eref{eq:thetadef1}.
The fifth line relies on \eref{eq:ldkplb}, the next on the fact
that $\Vert P(T_k,k)\Vert^2\ge 1/(2m)$ because $\Vert P(:,k)\Vert^2\ge 1/m$
(which follows from $\Vert P(:,k)\Vert_1=1$) and 
$\Vert P(\{1,\ldots,m\}-T_k,k)\Vert^2\le m\eps^2\le 3/(10m)$ from
\eref{eq:epsdef}.  The last line uses
$\phi\le 1/4$ because of \eref{eq:phidef2}.

Now we turn to the second term in \eref{eq:objfunc}.
Choose $\u,\v,\sigma$ so that 
$\u\sigma\v^T=P(T_k,k)\l(N)^T$, a rank-one matrix,
where as above $N=D_k'$.
Since $|A(i,j)-l_jP(i,k)|\le l_j\theta$ when $j$ is
acceptable, we have the following
bound for the second term.
\begin{eqnarray*}
\gamma\Vert A(M,N)-\u\sigma\v^T\Vert_F^2
&=& 
\gamma\sum_{j\in D_k'}\sum_{i\in T_k} (A(i,j)-l_jP(i,k))^2 \\
&\le &
\gamma\sum_{j\in D_k}\sum_{i\in T_k} l_j^2\theta^2 \\
&=&
\gamma\theta^2 \Vert \l(D_k)\Vert^2 \cdot |T_k|\\
&\le &
8\gamma \theta^2 mn_kL^2/(3q)
\\
&\le & n_kL^2/(256qm).
\end{eqnarray*}
The fourth line follows from \eref{eq:ldkub}, and the last follows
from \eref{eq:thetadef5} (taking $\gamma=4$).  
Thus, subtracting the above right-hand side
from \eref{eq:fopt1} shows that
$\fopt \ge n_kL^2/(256qm)$.
This inequality is valid for all $k=1,\ldots,t$, 
so we may assume it is true
for the $k$ that maximizes $n_k$.  This value of $n_k$ is therefore
at least $n/t$. This establishes \eref{eq:fopt}.
\end{proof}

For a particular topic
$k$, say that a term index $i\in T_k$ is {\em heavy} if
$P(i,k)\ge \chi$, where $\chi$ was defined by
\eref{eq:chidef1}--\eref{eq:chidef2} above.
Let $H_k\subset T_k$ be the heavy indices.
We use the notation $\top(j)$ to denote the topic of document $j$,
$j\in\{1,\ldots,n\}$.
Say that an entry $A(i,j)$ of $A(M,N)$ is a {\em heavy} entry
if $i\in H_k$ where $k=\top(j)$.
Finally, say that an entry $A(i,j)$ of $A(M,N)$
is {\em acceptable and heavy} if it is heavy and $j$ is
acceptable.

\begin{lemma}
Under Assumptions {\rm A1--A5},
the sum of squares of entries of $A$ that are not heavy but
are acceptable is at most $nL^2/(16\cdot 256\cdot 25^2 qtm).$
\label{lem:heavynotacc}
\end{lemma}

\begin{proof}
This is a straightforward estimate:
\begin{eqnarray}
\sum_{\mbox{\scriptsize $A(i,j)$ not heavy \& $j$ acceptable}} A(i,j)^2
&=&\sum_{k=1}^t\sum_{j\in D_k'}\sum_{i\notin H_k} A(i,j)^2 
\nonumber\\
&\le& \sum_{k=1}^t\sum_{j\in D_k'}\sum_{i\notin H_k} l_j^2(P(i,k)+\theta)^2 
\nonumber\\
&\le&\sum_{k=1}^t\sum_{j\in D_k'}\sum_{i\notin H_k} l_j^2(\chi+\theta)^2 
\nonumber\\
&=& (\chi+\theta)^2\sum_{k=1}^t(m-|H_k|)\sum_{j\in D_k'} l_j^2 
\nonumber\\
&\le & m(\chi+\theta)^2\Vert\l\Vert^2 
\nonumber\\
&\le & 8m(\chi+\theta)^2nL^2/(3q) 
\nonumber\\
&\le & 32m\chi^2nL^2/(3q) 
\label{eq:anotheavyacc0} \\
&\le & nL^2/(16\cdot256\cdot 25^2qtm).
\label{eq:anotheavyacc}
\end{eqnarray}
The second line follows by definition of `acceptable.'  The third
follows because $P(i,k)<\chi$ if $i$ is not heavy in topic $k$.
The sixth line follows from \eref{eq:lub}, the seventh because
$\theta\le \chi$ (refer to \eref{eq:thetadef6}) and the last
from \eref{eq:chidef1}.
\end{proof}

\begin{lemma}
Under Assumptions {\rm A1--A5},
The sum of squares of acceptable and
heavy entries in $A(M,N)$, where $M,N$ are the
optimizers of \eref{eq:objfunc}, is at least
$nL^2/(512qt)$.
\label{lem:heavyacc}
\end{lemma}

\begin{proof}
The sum of squares of entries in $A(M,N)$ from unacceptable
documents is bounded above by the sum of squares of entries in $A$ of
unacceptable documents, for which we have the estimate given by
\eref{eq:adunacc}.  The sum of squares
of entries of $A(M,N)$ which are acceptable
but not heavy is bounded above by the same quantity for all of $A$, which
is given by \eref{eq:anotheavyacc}.  Adding these two upper bounds
gives a quantity less than half of the lower bound in \eref{eq:fopt},
which proves the result.
\end{proof}

The following lemma is stated more generally than the others
 of this section 
(i.e., without Assumptions A1--A5 and without assuming
$\gamma=4$) because
it is more broadly applicable.  

\begin{lemma}
Let $A$ be an $m\times n$ matrix with nonnegative entries.
Let $M\subset\{1,\ldots,m\}$ be the optimizing choice of
$M$ for \eref{eq:objfunc}.  Assume
$\gamma>2$.
Let $j,j'$ index two columns of $A$ such that 
\begin{equation}
\frac{A(M,j)^TA(M,j')}{\Vert A(M,j)\Vert\cdot \Vert  A(M,j') \Vert}
< 1-2/\gamma.
\label{eq:bigangle}
\end{equation}
Then at least one of $j$ or $j'$ is not a member of the optimizing
choice of $N$.
\label{lem:mutualexclus}
\end{lemma}

\noindent{\bf Remark.} The lemma is also true when the roles of $M$
and $N$ are reversed since the value of the objective function
\eref{eq:objfunc2} is unchanged under matrix transposition.

\begin{proof}
Let unit vector
$\u$ be the optimizing choice for \eref{eq:objfunc}.  As noted
in Lemma~\ref{lem:separable}, $j$ and $j'$ are included in the
optimal $N$ provided $f_j,f_{j'}>0$, where
\begin{eqnarray*}
f_j &=& \gamma (A(M,j)^T\u))^2-(\gamma-1)\Vert A(M,j)\Vert^2, \\
f_{j'}&=&\gamma (A(M,j')^T\u))^2-(\gamma-1)\Vert A(M,j')\Vert^2.
\end{eqnarray*}
Here, we have simplified notation by allowing $\u$ to stand for $\u(M)$.
We will now show that for any possible choice of $\u$, either
$f_j<0$ or $f_{j'}<0$, meaning that at least
one of  $j$ or $j'$ cannot be in $N$.

To proceed, let us 
define normalizations
$\r=A(M,j)/\Vert A(M,j)\Vert$ and $\s=A(M,j')/\Vert A(M,j')\Vert$.
With these definition, \eref{eq:bigangle} is rewritten
$\r^T\s< 1-2/\gamma$.
Since multiplying by a positive scalar does not affect the signs of $f_j$
or $f_{j'}$, it suffices to redefine them using the normalized
vectors:
$$f_j = \gamma(\r^T\u)^2-\gamma+1$$
and
$$f_{j'} = \gamma(\s^T\u)^2-\gamma+1.$$
Thus,
\begin{eqnarray}
f_j+f_{j'}&= & \gamma(\r^T\u)^2+\gamma(\s^T\u)^2-2\gamma+2 \\
&=&
\gamma\left\Vert\left(
\begin{array}{c}
\r^T \\
\s^T
\end{array}
\right)
\u
\right\Vert^2 -2\gamma+2  \nonumber \\
&\le &
\gamma\left\Vert\left(
\begin{array}{c}
\r^T \\
\s^T
\end{array}
\right)
\right\Vert_2^2 -2\gamma+2 \nonumber \\
&=&
\gamma
\lambda_{\max}\left(
\begin{array}{cc}
1 & \r^T\s \\
\r^T\s & 1
\end{array}
\right)-2\gamma + 2. \label{eq:rts}
\end{eqnarray}
In this inequality we used the notation $\lambda_{\max}$ to denote
the maximum eigenvalue of a symmetric matrix.  We also used the identity
that for any matrix $B$, $\Vert B\Vert_2=(\lambda_{\max}(BB^T))^{1/2}$.
The eigenvalues of the $2\times 2$ matrix above can be easily
determined as  $1\pm \r^T\s$.   Thus, 
$$f_j+f_{j'}\le (1+\r^T\s)\gamma-2\gamma+2=(\r^T\s-1)\gamma+2.$$
Since $\r^T\s<1-2/\gamma$, the right-hand side is negative, thus
showing that either $f_j$ or $f_{j'}$ is negative.
\end{proof}

We can now apply the previous lemma to the text corpus under
analysis.

\begin{lemma}
Assume {\rm A1--A5} hold.
Suppose that $(i,j),(i',j')$ are two acceptable heavy entries
in the optimal solution $(M,N)$.  Then $(i,j)$ and $(i',j')$
must be from the same topic $k$.
\label{lem:uniqueheavy}
\end{lemma}

\begin{proof}
Suppose
that $(i,j)$ is an acceptable heavy entry on topic
$k$, and $(i',j')$ is an acceptable heavy  entry
topic $k'$ such that $k'\ne k$.  Suppose also that both $i,i'\in M$.
We will prove that either $j$ or $j'$ is not in $N$.
Let $\r=A(M,j)/\Vert A(M,j)\Vert$ and $\s=A(M,j')/\Vert A(M,j')\Vert$.
Let us split $\r$ and $\s$ into three subvectors: $\r_1,\s_1$
contain those entries indexed
$M\cap T_k$; $\r_2,\s_2$ contain
entries indexed by  $M\cap T_{k'}$; and $\r_3$, $\s_3$
contain the remaining
entries of $M$.  
Since $j\in T_k$, $(i,j)$ is heavy,
and $j$ is acceptable, this means that 
$i\in M\cap T_k$ and
$A(i,j)\ge l_j(\chi-\theta)$, so 
that $\Vert A(M\cap T_k,j)\Vert^2\ge l_j^2(\chi-\theta)^2.$
Since $\theta\le\chi/2$ (refer to \eref{eq:thetadef6}),
this quantity is at
least $l_j^2\chi^2/4$.  On the other hand, $\Vert A(M-T_k,j)\Vert^2
\le ml_j^2(\eps+\theta)^2$ since column $j$ is acceptable
and $P(i,k)\le\eps$ for $i\notin T_k$.  
Thus, after rescaling, 
\begin{eqnarray*}
\Vert[\r_2;\r_3]\Vert & = & 
\Vert\r(M-T_k)\Vert \\
&=& \Vert A(M-T_k,j)\Vert/\Vert A(M,j)\Vert \\
& \le & \Vert A(M-T_k,j)\Vert / \Vert A(M\cap T_k,j)\Vert \\
& \le & 2m^{1/2}(\eps+\theta)/\chi.
\end{eqnarray*}
The inequality $\chi\ge 10m^{1/2}(\eps+\theta)$
follows from   \eref{eq:thetadef6} and
the fact that $\eps\le \theta$ from \eref{eq:epsdef}.
Thus,
$\Vert [\r_2;\r_3]\Vert\le 1/5$.
Similarly, $\Vert[\s_1;\s_3]\Vert\le 1/5$.  Hence,
\begin{eqnarray*}
|\r^T\s| &\le& |\r_1^T\s_1|+|\r_2^T\s_2|+|\r_3^T\s_3| \\
&\le &\Vert \s_1\Vert + \Vert \r_2\Vert + \Vert\r_3\Vert\cdot\Vert \s_3\Vert \\
&\le & 1/5+1/5+1/25 =0.44.
\end{eqnarray*}
Thus, by Lemma~\ref{lem:mutualexclus}, since $0.44<1-2/\gamma$ when $\gamma=4$,
either $j$ or $j'$ is not present in the optimal $N$.
\end{proof}

The next lemma shows that the left and right singular vectors of the
optimal solution to \eref{eq:objfunc} are determined largely by
the document lengths and probability distribution for topic $k$.

\begin{lemma}
Let $(M,N)$ be the index sets that optimize \eref{eq:objfunc}.
Let $k$ be the index of the topic of the heavy entry occurring
in $(M,N)$ (which is uniquely determined according to 
Lemma~$\ref{lem:uniqueheavy}$).  Let $\r$ and $\s$ be the right and 
left singular vectors of
$A(M,N)$, respectively.  Assume {\rm A1--A5} hold.
Then there exists a positive scalar $\kappa_r$ such that
\begin{equation}
\frac{\Vert \kappa_r\r-\r_0\Vert}{\Vert\r_0\Vert}\le 1/6,
\label{eq:kappar}
\end{equation}
where $\r_0$ is defined by
\begin{equation}
r_0(j)=\left\{
\begin{array}{ll}
l_j & \mbox{for $j\in D_k\cap N$,} \\
0 & \mbox{else.}
\end{array}
\right.
\label{eq:r0def}
\end{equation}
Similarly, there exists a positive scalar $\kappa_s$ such that
\begin{equation}
\frac{\Vert \kappa_s\s-\s_0\Vert}{\Vert\s_0\Vert}\le 1/6,
\label{eq:kappas}
\end{equation}
where $\s_0=P(M,k)$.
In addition, $\r_0$ satisfies the following inequality:
\begin{equation}
\Vert\r_0\Vert^2 \ge nL^2/(278qtm)
\label{eq:r0lb}
\end{equation}
\label{lem:approxsingvec}
\end{lemma}

\begin{proof}
Let $B$ be a matrix
with the same dimensions as $A(M,N)$, indexed the same way
$A$ is indexed (i.e., $B$ and $B(M,N)$
denote the same matrix) and whose $(i,j)$ entry is $B(i,j)=P(i,k)r_0(j)$.
We will also use $\r_0$ and $\r_0(N)$
synonymously, and similarly for $\r$, $\s$ and $\s_0$.
Observe that $B$ is a rank-one
matrix since $B=P(M,k)\r_0^T$.
The only nonzero singular value of $B$
is $\Vert \r_0\Vert\cdot \Vert P(M,k)\Vert$.

Let us partition $N$ into three sets $N_1\cup N_2\cup N_3$ given by
\begin{equation}
N_1=N\cap D_{\rm unacc};\quad N_2=N\cap(D_{\rm acc}-D_k');\quad
N_3=N\cap D_k'.
\label{eq:Npart}
\end{equation}
From this partition,
\begin{eqnarray*}
\Vert B-A(M,N)\Vert_2^2 
& \le &
\Vert B-A(M,N)\Vert_F^2 
\\
& = &
\Vert B(M,N_1)-A(M,N_1)\Vert_F^2 
+\Vert B(M,N_2)-A(M,N_2)\Vert_F^2 
\\
& & \quad\mbox{}+
\Vert B(M,N_3)-A(M,N_3)\Vert_F^2
\end{eqnarray*}
where we now obtain upper bounds on
the three terms individually.
\begin{eqnarray*}
\Vert B(M,N_1)-A(M,N_1)\Vert_F^2 &\le &
\Vert B(M,N_1)\Vert_F^2 + \Vert A(M,N_1)\Vert_F^2 
\\
& = & \Vert P(M,k)\Vert^2\Vert \l(N_1\cap D_k)\Vert^2+
\sum_{j\in N_1}\sum_{i\in M} A(i,j)^2
\\
&\le & \Vert \l(D_{\rm unacc})\Vert^2 +\Vert A(:,D_{\rm unacc})\Vert_F^2
\\
&\le &
2nL^2/(256\cdot 3\cdot2\cdot25^2qtm).
\end{eqnarray*}
In the above derivation, we used \eref{eq:aminusb} for
the first line, the relationship $N_1\subset D_{\rm unacc}$
and $\Vert P(:,k)\Vert\le 1$ for the third line, and
\eref{eq:ldunacc} and \eref{eq:adunacc} for the last.

Next, observe that $B(M,N_2)=0$ since $N_2\cap D_k=\emptyset$, hence
$$\Vert B(M,N_2)-A(M,N_2)\Vert_F^2  = \Vert A(M,N_2)\Vert_F^2.$$
All entries of the right-hand side are acceptable
and not heavy; they are acceptable by choice of $N_2$, and they are
not heavy because Lemma~\ref{lem:uniqueheavy} shows that there
cannot be an acceptable heavy entry from a topic other than $k$
in the optimal solution.  Thus, from \eref{eq:anotheavyacc},
$$\Vert B(M,N_2)-A(M,N_2)\Vert_F^2 \le  nL^2/(16\cdot256\cdot 25^2qtm).$$

Finally, for $j\in N_3$, $r_0(j)=l_j$ since $N_3\subset D_k$.  Thus,
\begin{eqnarray*}
\Vert B(M,N_3)-A(M,N_3)\Vert_F^2
&=& 
\sum_{j\in N_3}\sum_{i\in M} (l_jP(i,k)-A(i,j))^2
 \\
&\le&
\sum_{j\in N_3}\sum_{i\in M} l_j^2\theta^2
\\
&=&
|M|\theta^2\Vert \l(N_3)\Vert^2
\\
&\le&
m\theta^2\Vert \l(D_k)\Vert^2
\\
&\le &
8m\theta^2n_kL^2/(3q)
\\
&\le &
nL^2/(3\cdot 25^2\cdot 256qtm)
\end{eqnarray*}
Here, the definition of `acceptable' was used for the second line, and
\eref{eq:ldkub} was used for the fifth line and \eref{eq:thetadef5}
for the last line.

Thus, we see that $\Vert B(M,N)-A(M,N)\Vert^2\le nL^2/(25^2\cdot 256qtm)$.
On the other hand, $\Vert A(M,N)\Vert^2\ge nL^2/(256 qt)$
by \eref{eq:fopt}.  Thus
$\Vert B(M,N)-A(M,N)\Vert/\Vert A(M,N)\Vert\le 1/25$.  This
means by the triangle inequality that 
$\Vert B(M,N)-A(M,N)\Vert/\Vert B(M,N)\Vert\le 1/24$. 

Now we 
can apply Theorem 8.6.5 of Golub and Van Loan \cite{GVL}
on the perturbation of singular vectors to conclude that 
the normalized left singular vector of $A(M,N)$ differs from
$\r_0/\Vert \r_0\Vert$ by at most 
$1/6$.  (Note that in applying the theorem, we use the fact
that the second singular value of $B(M,N)$ is zero since $B(M,N)$
has rank one.)
Similarly,
the normalized right singular vector of $A(M,N)$ differs
from $\s_0/\Vert \s_0\Vert$ by at most $1/6$.

Finally, to establish
\eref{eq:r0lb}, we observe that
\begin{eqnarray*}
\Vert \r_0\Vert^2 &=&\Vert B\Vert^2/\Vert P(M,k)\Vert^2 \\
&\ge & (24/25)^2nL^2/(256qtm).
\end{eqnarray*}
which implies \eref{eq:r0lb}.
The first line follows because $B$ is rank-one, and the second
because $\Vert B\Vert\ge (24/25)\Vert A(M,N)\Vert$ as established
above, and $\Vert P(:,k)\Vert\le1$ since $\Vert P(:,k)\Vert_1=1$.
\end{proof}

The following lemma is used to determine when adding a row or
column to the sets $M$ or $N$ will increase the objective
function \eref{eq:objfunc}.

\begin{lemma}
Let $\u$ be a nonzero vector and $\d_1,\d_2$ perturbations such
that $\Vert \d_1\Vert\le \Vert\u\Vert/6$ and 
$\Vert\d_2\Vert\le \Vert\u\Vert/6$.
Suppose $\a=\kappa_1(\u+\d_1)$ and $\b=\kappa_2(\u+\d_2)$,
where $\kappa_1,\kappa_2$ are positive scalars.
Then 
$$\a^T\a-\gamma \Vert \a-\beta\b\Vert^2>0$$
for $\gamma=4$
and for at least one choice of $\beta$.
\label{lem:incrobj}
\end{lemma}
\begin{proof}
Let us take $\beta=\kappa_1/\kappa_2$.  Then 
\begin{eqnarray*}
\a^T\a-\gamma\Vert\a-\beta\b\Vert^2 &=&
\kappa_1^2(\u+\d_1)^T(\u+\d_1)-\gamma\Vert\kappa_1(\u+\d_1)-
\kappa_1(\u+\d_2)\Vert^2 \\
&=&\kappa_1^2\left[\u^T\u + 2\u^T\d_1+(1-\gamma)\d_1^T\d_1
-2\gamma\d_1^T\d_2-\gamma\d_2^T\d_2\right] \\
&=&\kappa_1^2\left[\u^T\u + 2\u^T\d_1-3\d_1^T\d_1
-8\d_1^T\d_2-4\d_2^T\d_2\right] \\
&\ge&\kappa_1^2\left[\Vert\u\Vert^2-2\Vert\u\Vert\cdot\Vert\d_1\Vert
-3\Vert\d_1\Vert^2-8\Vert\d_1\Vert\cdot\Vert\d_2\Vert -4\Vert\d_2\Vert^2\right]
\\
&\ge &\kappa_1^2\Vert\u\Vert^2(1-2/6-3/36-8/36-4/36) \\
&> & 0.
\end{eqnarray*}
\end{proof}
The point of Lemma~\ref{lem:incrobj} is as follows.  Suppose $(M,N)$ is
a putative solution for maximizing \eref{eq:objfunc} and $j\notin N$.
Suppose the right singular vector of $A(M,N)$ is $\s$, and suppose
that $A(M,j)$ and $\s$ are both within relative distance of $1/6$ from
another vector $P(M,k)$ after rescaling, 
where $k$ is the topic of column $j$.  
Recall that once $M$ and $\u$ are fixed, the objective function 
of \eref{eq:objfunc} becomes separable by columns, i.e., it is
possible to choose the $j$th
entry of $\v$ considering only the contribution
of column $j$ to the total objective function.
The previous lemma says that there is a way to choose $v_j$
so that $(M,N\cup\{j\})$ has a higher objective function value
than $(M,N)$, where we take the same $M,\u,\sigma$ and extend $\v$ with the
particular choice of $v_j$.   
This means that in fact $N$ is not optimal, because it should
also include $j$.  The lemma can also be used on rows using the
analogous argument.

Now let us apply this lemma to deduce the contents of the
optimal $M$ and $N$.

\begin{lemma}
Assume {\rm A1--A5} hold.
In the optimal solution, $D_k'\subset N$.
\label{lem:dksubsN}
\end{lemma}

\begin{proof}
Let us consider a column $j\in D_k'$, that is, an
acceptable column for topic $k$. Observe that $A(M,j)=l_j(P(M,k)+\d_2)$,
where each entry of $\d_2$ has absolute value at most $\theta$
by definition of `acceptable.'  Now we observe that
$\Vert\d_2\Vert\le \theta m^{1/2}\le\chi/6\le \Vert P(M,k)\Vert/6$;
the first follows because $\Vert \d_2\Vert_\infty\le \theta$,
the second follows from \eref{eq:thetadef6}, and
the  third follows because $M$ contains at least one heavy row of $k$.

Thus, $A(M,j)=l_j(P(M,k)+\d_2)$ with $\d_2\le \Vert P(M,k)\Vert/6$
and the left singular vector $\s$ of $A(M,j)$ satisfies 
$\s=(P(M,k)+\d_1)/\kappa_s$, with $\Vert \d_1\Vert_2\le \Vert P(M,k)\Vert/6$
by \eref{eq:kappas}.

Thus, by Lemma~\ref{lem:incrobj}, column $j\in D_k'$ increases
the value of the objective function since $A(M,j)$ and the left
singular value of $A(M,N)$ are both scalar multiplies of perturbations
of $P(M,k)$, where the relative perturbation size is at most $1/6$.
This proves that all columns of $D_k'$ will lie in $N$.
\end{proof}

Recall that $H_k$ denotes the subset of $T_k$ of heavy rows (terms)
associated with topic $k$.

\begin{lemma}
Assume {\rm A1--A5} hold.
In the optimal solution, $H_k\subset M$.
\label{lem:hksubsM}
\end{lemma}

\begin{proof}
Let us consider a row $i\in H_k$.
Let us write $A(i,N)=P(i,k)(\r_0+\d_2)$ and
try to estimate $\d_2$.  By definition of $H_k$, $P(i,k)\ge\chi$.
We can obtain an upper bound on
$\d_2=A(i,N)/P(i,k)-\r_0(N)$
as follows.  Use the partition of $N$ given by
\eref{eq:Npart}.
Then 
\begin{eqnarray*}
\Vert \d_2(N_1)\Vert^2 &\le & \Vert A(i,N_1)\Vert^2/P(i,k)^2 + \Vert\r_0(N_1)\Vert^2
\\
&\le &
(1/\chi^2)\Vert\l(N_1)\Vert^2 + \Vert \l(N_1)\Vert^2 
\\
&\le&
(1+1/\chi^2)\Vert \l(D_{\rm unacc})\Vert^2 
\\
&\le & 
32\phi(1+1/\chi^2)nL^2/(3q),
\end{eqnarray*}
using \eref{eq:aminusb} for the first line,
\eref{eq:r0def},
\eref{eq:ajltlj} and $P(i,k)\ge \chi$ for the second line,
$N_1\subset D_{\rm unacc}$ for the third,
\eref{eq:ldunacc1} for the fourth.

Next,
\begin{eqnarray*}
\Vert \d_2(N_2)\Vert^2 &= & \Vert A(i,N_2)/P(i,k)^2 -\r_0(N_2)\Vert^2
\\
&=& \Vert A(i,N_2) \Vert^2/P(i,k)^2
\\
&= & (1/P(i,k)^2)\sum_{j\in N_2} A(i,j)^2 
\\
&\le & (1/P(i,k)^2)\sum_{j\in N_2} l_j^2(P(i,\top(j))+\theta)^2
\\
&\le &
(1/\chi^2)\sum_{j\in N_2} l_j^2(\epsilon+\theta)^2 
\\
&\le &
((\epsilon+\theta)/\chi)^2 \Vert \l\Vert^2 
\\
&\le &
8((\epsilon+\theta)/\chi)^2nL^2/(3q).
\end{eqnarray*}
For the second line we used the fact that $\r_0(N_2)=\bz$, which follows
from \eref{eq:r0def} and \eref{eq:Npart}.  For the fourth line we
used the fact that $j$ is acceptable.  For the fifth we used
$P(i,k)\ge\chi$ and
$P(i,\top(j))\le \eps$ since
$i\in H_k\subset T_k$ and $j\notin D_k$.
For the last line we used \eref{eq:lub}.

Finally,
\begin{eqnarray*}
\Vert \d_2(N_3)\Vert^2 &=& 
\sum_{j\in N_3} (A(i,j)/P(i,k)-l_j)^2 
\\
&\le & 
\sum_{j\in N_3} l_j^2\theta^2 
\\
&\le & \theta^2 \Vert \l(D_k)\Vert^2
\\
& \le & 8\theta^2n_kL^2/(3q).
\end{eqnarray*}
The second line follows because $N_3\subset D_{\rm acc}$.  The last
line follows from \eref{eq:ldkub}.
Thus,
$$\Vert \d_2\Vert \le 
(1+1/\chi^2)\frac{32\phi nL^2}
{3q}
+
\frac{8(\eps+\theta)^2nL^2}{\chi^2\cdot 3q}
+
\frac{8\theta nL^2}{3q}
.$$
Now apply 
\eref{eq:phidef3} to the first term (plus the 
fact $(1+1/\chi^2)\le 2/\chi^2$),
\eref{eq:thetadef6}
to the second and 
\eref{eq:thetadef5} to the third term
to conclude that
$$\Vert \d_2\Vert \le \frac{nL^2}{6\cdot 278\cdot qmt}.$$
Comparing to \eref{eq:r0lb}, $\Vert \d_2\Vert \le \Vert \r_0\Vert/6$.
Thus, $A(i,N)$ is a perturbation of $\r_0$ of relative size at most
$1/6.$  By \eref{eq:kappar}, the right singular vector of $A(M,N)$ is
also a perturbation of $\r_0$ of relative size at most $1/6$.  By
Lemma~\ref{lem:incrobj}, inserting $i$ into $M$ can only increase the
objective function value.
\end{proof}

Now finally we can prove Theorem~\ref{thm:main}.
\begin{proof}
Consider an entry $(i,j)\in (T_k\times D_k)\bigtriangleup(M\times N)$.
We take two cases: the first case is $(i,j)\in (T_k\times D_k)-(M\times N)$.
In this case, since $H_k\subset M$ and $D_k'\subset N$ as proved in the two
preceding lemmas, it must be the case that either $j\in D_k-D_k'$ or
$i\in T_k-H_k$, i.e., either the entry is in an unacceptable column or
it is a acceptable but not heavy.

The second case is $(i,j)\in (M\times N)-(T_k\times D_k)$.  
Thus, either $j$ is on a topic other than $k$ (i.e., $j\notin D_k$),
or it is on topic $k$ but
is not a heavy entry (i.e., $j\in D_k$ but $i\notin T_k$, so 
$i\notin H_k$).
Thus, either $j$ is
an unacceptable column, or else $i$ is not a heavy entry, because if
by Lemma~\ref{lem:uniqueheavy}, $A(M,N)$ cannot contain any acceptable
and heavy entries except on topic $k$.

Thus, we see that all entries indexed by
$(T_k\times D_k)\bigtriangleup(M\times N)$ are either unacceptable or
not heavy.  The maximum norm of unacceptable entries is given by
\begin{eqnarray*}
\Vert A(:,D_{\rm unacc})\Vert_F^2 &\le& 32\phi nL^2/(3q) \\
&\le &  \alpha nL^2/(512qtm),
\end{eqnarray*}
where the first line comes from \eref{eq:adunacc0} and the
second from \eref{eq:phidef4}.

The maximum norm of entries that are acceptable but not heavy is
\begin{eqnarray*}
\sum_{\mbox{\scriptsize $A(i,j)$ not heavy \& $j$ acceptable}} A(i,j)^2
&\le & 32m\chi^2nL^2/(3q) 
\\
&\le & \alpha nL^2/(512qtm),
\end{eqnarray*}
where the first line comes from \eref{eq:anotheavyacc0} 
and the second from \eref{eq:chidef2}.
Thus, adding the two previous inequalities shows that the sum of
entries indexed by the symmetric difference 
$(T_k\times D_k)\bigtriangleup(M\times N)$ is at most $\alpha nL^2/(256qtm)$.
This is a fraction of at most $\alpha$ times the optimal value as
shown by \eref{eq:fopt}.
\end{proof}

\section{Behavior of the objective function on decomposable bitmap images}
\label{sec:imagetheory}

We consider the behavior of objective function \eref{eq:objfunc} on
{\em decomposable bitmap images}.  A {\em bitmap image} is
one in which each pixel is either white (0) or black (1).
Suppose that $A$ is an $m\times n$ matrix encoding a family of
images; here $m$ is the number of pixels per image and $n$ is the
number of images.  Since the images are assumed to be bitmaps, every
entry of $A$ is either 0 or 1.

A collection
of images  is {\em decomposable} if there exists a partitioning of the pixel
positions $\{1,\ldots,m\}$ 
into $t$ subsets  $T_1,\ldots,T_t$, called {\em features},
such that for every $k$, every image is
either black in all of $T_k$ or is white in all of $T_k$.
Clearly
any collection of images is decomposable into
individual pixels (i.e., $T_1=\{1\}$, $T_2=\{2\}$, etc.), so the interesting
case is when $t\ll m$.  Donoho and Stodden
\cite{donoho} have considered a particular kind
of decomposable bitmap image database.

We now consider a simple probabilistic model of generating a database
of $n$ decomposable bitmap images (i.e., an $m\times n$ matrix)
and prove that the objective function
\eref{eq:objfunc} is able to identify a feature in the database with
high probability.  There are many other ways to define a model
for which a similar theorem could be proved.

\begin{theorem}
Let $T_1,\ldots,T_t$, the features,
be a partition of $\{1,\ldots,m\}$ with $t>1$.
Let $m_{\min},m_{\max}$ denote $\min_{k=1,\ldots,t}|T_k|$, $\max_{k=1,\ldots,t}
|T_k|$ respectively.
Let $l$ be an integer in $1,\ldots,t/2$.
Assume that each of the $n$ images in the matrix $A$ is generated 
independently by selecting exactly
$l$ features uniformly at random out of the possible $t$.
Finally, assume that $\gamma>4m/m_{min}$.  Then
with probability tending to $1$ as $n\rightarrow\infty$, 
the optimizer of \eref{eq:objfunc} applied to this $A$ 
will select $M=T_k$ for some $k$ such that $|T_k|=m_{\max}$.
\end{theorem}

\begin{proof}
Let $(M,N)$ be the optimizing solution of \eref{eq:objfunc}.  
Observe that, for any $k$,
all the rows of $A$ indexed by $T_k$ are identical
by construction.  Therefore, it follow from \eref{eq:gidef} that
if any row from $T_k$ lies in $M$, then all of $T_k$ must
be included in $M$ since all have the same $g_i$ value.
Thus, $M$ is a union of some of the $T_k$'s.

We claim that it is impossible that 
the optimal $N$ contains two columns $j$ and $j'$ such that bitmap
$j$ contains feature $T_k$ for $T_k\subset M$
while $j'$ does not contain $T_k$.
The reason is that in this case, $A(M,j)$ consists of a vector 
with $1$'s in positions indexed by $T_k$ and while $A(M,j')$
has $0$'s in these positions.   Let $m_1$ be the number of `1' pixels
in $j$ and $m_2$ be the number of `1' pixels in $j'$.
Observe $m_1\le m$ and $m_2\le m$.
Let $q=|T_k|$ so that $q\ge m_{\min}.$
Then the number of `1' pixels
in common between images $j$ and $j'$ is at most $\min(m_1-q,m_2)$.
Consider the left-hand
side of \eref{eq:bigangle}:
\begin{eqnarray}
\frac{A(M,j)^TA(M,j')}{\Vert A(M,j)\Vert\cdot\Vert A(M,j')\Vert}
&\le &\frac{\min(m_1-q,m_2)}{\sqrt{m_1m_2}} \nonumber \\
&\le &\left \{
\begin{array}{ll}
\frac{m_1-q}{\sqrt{m_1m_2}}, &  \mbox{if $m_1-q\le m_2$}, \\
\frac{m_2}{\sqrt{m_1m_2}}, &  \mbox{if $m_1-q\ge m_2$}
\end{array}
\right. \nonumber \\
& \le &
\left \{
\begin{array}{ll}
\frac{m_1-q}{\sqrt{m_1(m_1-q)}}, &  \mbox{if $m_1-q\le m_2$}, \\
\frac{\sqrt{m_1-q}}{\sqrt{m_1}}, &  \mbox{if $m_1-q\ge m_2$}
\end{array}
\right. \nonumber \\
& = & \sqrt{1-q/m_1} \nonumber \\
& \le & \sqrt{1-m_{\min}/m} \nonumber \\
& \le & 1-m_{\min}/(2m). \label{eq:imangbound}
\end{eqnarray}
By assumption, $\gamma>4m/m_{\min}$, so the right-hand side
of \eref{eq:imangbound} is less than $1-2/\gamma$.
Thus, by Lemma~\ref{lem:mutualexclus},
not both $j$ and $j'$ can be in the optimal choice of $N$.

Thus, we conclude that all columns taking part in the optimal solution
must have all 1's (or all 0's) in positions indexed by $M$.  Ignore
the columns of all 0's since their presence does not affect the
objective function value.  Consider now a feature $k$ such that
$|T_k|=m_{\max}$.  Feature $k$ is expected to occur in
the fraction $l/t$ of columns of $A$.  
For any $\eps>0$, by choosing $n$ sufficiently
large, we can assume with probability arbitrarily close to 1
that this choice occurs in the fraction at least $l/t-\eps$ of the columns.
Therefore, 
\begin{equation}
f(T_k,N)\ge n(l/t-\eps)m_{\max},
\label{eq:T_kval}
\end{equation}
for any $\eps>0$ and $n$ sufficiently large,
where $N$ is the set of
columns containing feature $k$.

Now consider any other possible choice of $M$; suppose e.g., that
$M$ has $s$ of the features.
By the preceding argument, the optimal choice
of $N$ that could accompany this $M$ contains only columns that
use  all $s$ features.
This union
of $s$ features 
is expected to occur in the fraction 
$$\frac{l(l-1)\cdots(l-s+1)}{t(t-1)\cdots(t-s+1)}$$
of the columns.  
Thus, for any $\eps>0$, for $n$ sufficiently large, 
\begin{equation}
f(M,N) \le n\cdot \frac{l(l-1)\cdots(l-s+1)+\eps}{t(t-1)\cdots(t-s+1)}\cdot sm_{\max}.
\label{eq:fMN}
\end{equation}
(The factor $sm_{\max}$ is the maximum
contribution to $\Vert A(M,N)\Vert_F$ from
a particular column of $N$.)  Now it is a simple matter to check that
for any positive integers $l$, $t$ such that $l\le t/2$ and $s\le l$,
$$\frac{l(l-1)\cdots(l-s+1))}{t(t-1)\cdots(t-s+1)}\cdot s
\le \frac{l}{t}$$
with strict inequality for $s>1$.  Thus, by comparing
\eref{eq:T_kval} with \eref{eq:fMN}, we conclude that
as $n\rightarrow\infty$, 
with probability tending to 1, $f(T_k,N)$ is the optimal value
of the objective function.
\end{proof}

It should be noted that the previous theorem states that the optimal
$M$ includes a single feature $k$ but says nothing about the optimal $N$.
Indeed, as noted in the proof, we can take the optimal $N$ to
be $\{1,\ldots,n\}$.  In some situations it might be desirable for the
optimal $N$ to include only those images that use feature $k$.
This can be achieved by including a penalty term \eref{eq:secondpenalty}
into the objective function in which $\rho$ is chosen to be a very
small positive coefficient.

\section{On the complexity of maximizing $f(M,N)$}
\label{sec:nphard}

In this section, we observe that the problem of globally maximizing
\eref{eq:objfunc2} is NP-hard at least in the case that $\gamma$
is treated as an input parameter.  This observation explains why
R1D settles for a heuristic maximization of \eref{eq:objfunc2} rather
than exact maximization.
First, observe that the
maximum biclique (MBC) problem is NP-hard as proved by Peeters \cite{Peeters}.
We show that the MBC problem can
be transformed to an instance of \eref{eq:objfunc2}.

Let us recall the definition of the MBC problem.  The input is
a bipartite graph $G$.  The problem is to find an $(m,n)$-complete bipartite
subgraph $K$ (sometimes called a {\em biclique}) 
of $G$ such that $mn$ is maximized, i.e., the number of edges of $K$ is
maximized.

Suppose we are given $G$, an instance of the maximum biclique problem.
Let $A$ be the left-right adjacency matrix of $G$, that is, if $G=(U,V,E)$
where $U\cup V$ is the bipartition of the node set, then $A$ has
$|U|$ rows and $|V|$ columns, and $A(i,j)=1$ if $(i,j)\in E$ for $i\in U$
and $j\in V$, else $A(i,j)=0$.

Consider maximizing \eref{eq:objfunc2} for this choice of $A$.  We
require the following preliminary linear-algebraic lemma.

\begin{lemma}
Let $A$ be a matrix that has either of the
following as a submatrix:
\begin{equation}
U_1=\left(
\begin{array}{cc}
1  & 0 \\
0 & 1
\end{array}
\right)
\mbox{ or }
U_2=\left(
\begin{array}{cc}
1  & 1 \\
0 & 1
\end{array}
\right).
\label{eq:u1u2}
\end{equation}
Then $\sigma_2(A)>0.618$.
\end{lemma}
\begin{proof}
First,
observe that if $U$ is a submatrix of $A$, then $\Vert U\Vert_2\le
\Vert A\Vert_2$.  This follows directly from the operator
definition of the matrix 2-norm. 

Next, recall the following preliminary well known
fact (Theorem 2.5.3 of \cite{GVL}): 
For any matrix $A$,
$\sigma_i(A)=\min\{\Vert A-B\Vert_2: \rank(B)=i-1\}$.  
This
fact implies the following
generalization of the result in the
previous
paragraph: if $U$ is a submatrix of $A$, say $U=A(M,N)$,
then for any $i$,
$\sigma_i(U)\le \sigma_i(A)$.  
The reason is
that if $\tilde B\in\R^{|M|\times |N|}$ is a rank-$k$ matrix, 
then
$B\in\R^{m\times n}$ defined by padding with zeros is also rank
$k$, and $\Vert A-B\Vert_2\ge \Vert U-\tilde B\Vert_2$ by the
result in the previous paragraph. 
Now finally the lemma is proved, since $\sigma_2(U_1)=1$ and
$\sigma_2(U_2)>0.618.$
\end{proof}

This lemma leads to the following lemma.
\begin{lemma}
Suppose all entries of
$A\in \R^{m\times n}$ are either $0$ or $1$, and suppose
and at least one entry is 1.  Suppose $M,N$ are the optimal solution
for maximizing $f(M,N)$ given by \eref{eq:objfunc2}.  Suppose also that
the parameter $\gamma$ is chosen to be $2.7mn+1$ or larger.
Then the optimal choice of $M,N$ must yield a matrix $A(M,N)$
of all $1$'s, possibly
augmented with some rows or columns that are entirely zeros.
\end{lemma}

\begin{proof}
First, note that the optimal objective function value is at least 1
since we could take $M=\{i\}$ and $N=\{j\}$ where $(i,j)$ are chosen
so that $A(i,j)=1$.  In this case, $f(M,N)=1$.

Let $(M,N)$ be a pair of index sets that is a putative optimum for
\eref{eq:objfunc2}.
Suppose $A(i,j)=0$, where $(i,j)\in M\times N$.  
One possibility is that either row $i$ or column $j$ is entirely
made of 0's, in which case $(M,N)$ conforms to the claim in the lemma.
The other case is that $A(i,j)=0$
and yet there is an $i'\in M$ and $j'\in N$ such that $A(i',j)=A(i,j')=1$.
In this case, $A(M,N)$ has as a submatrix (using rows $\{i',i\}$ and
columns $\{j',j\}$) one of the two special matrices $U_1$ or $U_2$ from
\eref{eq:u1u2}.  Thus, $\sigma_2(A(M,N))\ge 0.618$.  On the other hand,
$\sigma_1(A(M,N))^2\le \Vert A(M,N)\Vert_F^2\le mn$.  Therefore,
$f(M,N)\le mn-(\gamma-1) (0.618^2)\le 0$ since $(\gamma-1)(0.618)> mn$
by choice of $\gamma.$  In particular, this means
$(M,N)$ cannot be optimal.
\end{proof}

If $A(M,N)$ includes a row or column entirely of zeros, then this row
or column may be dropped without affecting the value of the objective
function \eref{eq:objfunc2}.  Hence it follows
from the lemma that without loss of generality that the
optimizer $(M,N)$ of \eref{eq:objfunc2} indexes a matrix of all $1$'s.
In that case,
$\sigma_1(A(M,N))=\sqrt{|M|\cdot|N|}$ while $\sigma_2(A(M,N))=\cdots=
\sigma_p(A(M,N))=0$ (where $p=\min(|M|,|N|)$), and hence
$f(M,N)=|M|\cdot|N|$.  Thus, the value of the objective function 
corresponds exactly to the number of edges in the biclique.
This completes the proof that biclique is reducible in polynomial time
to maximizing \eref{eq:objfunc2}.

We note that Gillis \cite{gillis}
also uses the result of Peeters for a similar purpose,
namely, to show that the subproblem arising in
his NMF algorithm is
also NP-hard.

The NP-hardness result in this section requires that $\gamma$ be
an input parameter.  We conjecture that \eref{eq:objfunc2} is 
NP-hard even when $\gamma$ is fixed (say $\gamma=4$ as used herein).

\section{Image database test cases}
\label{sec:image}

\section{Text database test cases}
\label{sec:text}

\section{Conclusions}
We have proposed an algorithm called R1D for nonnegative matrix
factorization.  It is based on greedy rank-one downdating according
to an objective function, which is heuristically maximized.
We have shown that the objective function is well suited for
identifying topics in the $\epsilon$-separable text model
and on a model of decomposable bitmap images.
Finally, we have shown that the algorithm performs well in practice.

This work raises several interesting open questions.  First, the
$\epsilon$-separable text model seems rather too simple to describe
real text, so it would be interesting to see if the results generalize
to more realistic models.

One straightforward generalization is to consider power-law models
for document lengths that generalize the Zipf law: suppose that the
probability of length $l$ occurring is proportional to $l^{-p}$ for some
$p$.  Our proof of Theorem~\ref{thm:main}
generalizes to cover the case $0\le p <1$ without too much difficulty.
However, proving the theorem in the case $p>1$ appears to be much more
difficult (and perhaps the theorem is not true in this case).  When $p>1$,
short documents dominate the corpus, and short documents are not
easily analyzed using Chernoff-Hoeffding bounds.

A second question is to generalize the model of
image databases for which a theorem can be established.

A third question asks
whether a result like Theorem~\ref{thm:main}
will hold for the R1D algorithm using our proposed
definition of
the heuristic
subroutine {\tt ApproxRankOneSubmatrix}.
When {\tt ApproxRankOneSubmatrix}.
is applied to an $\eps$-separable corpus, does it successfully identify
a topic?  Here is an example
of a difficulty.  Suppose $n\rightarrow\infty$
much faster than $L$.  In this case, the document $j$ with the highest
norm will be the one in which $l_j$ is very close to $L$ and in which one
entry $A(i,j)$ is very close to $L$ while the rest are mostly zeros.
This is because 
the maximizer of $\Vert \x\Vert_2$ subject to the constraint
that $\Vert\x\Vert_1=C$ occurs when one entry of $\x$ is equal to $C$
and the rest are zero. 
It is likely
that at least one instance of a such a document will occur regardless
of the matrix $P(\cdot,\cdot)$ if $n$ is sufficiently large.  This document
will then act as the seed for expanding $M$ and $N$, but it may not
be similar to any topic.   

The scenario described in the preceding paragraph 
can apparently be prevented by
requiring $n$ and $L$ to grow proportionately, but the analysis
appears to be complicated in this case.  If we assume that the initial
column $j$
selected by R1D is well approximated by $l_jP(:,k)$ for some $k$
(i.e., the column is `acceptable' in the terminology of 
Theorem~\ref{thm:main}),
then the rest of $M$ and $N$ is likely to be the topics and terms associated
with topic $k$.  This is because
Lemma~\ref{lem:uniqueheavy} indicates that it is unlikely that a column
or row associated to another topic can improve the objective function
\eref{eq:objfunc}, whereas Lemmas~\ref{lem:dksubsN} and 
\ref{lem:hksubsM} indicate
that a document or term associated with topic $k$ will be favored
by \eref{eq:objfunc}.

\bibliography{../../Bibfiles/nips1,../../Bibfiles/nips}
\bibliographystyle{plain}
\end{document}